\documentclass[11pt,a4paper,reqno]{amsart}
\usepackage{amsthm,amsmath,amsfonts,amssymb,amsxtra,appendix,bookmark,dsfont,latexsym,bm,hyperref,delarray,color,euscript,amsgen,amsbsy,amsopn,amscd,latexsym,mathrsfs}

\makeatletter

\usepackage{hyperref}

\makeatletter
\def\@tocline#1#2#3#4#5#6#7{\relax
  \ifnum #1>\c@tocdepth 
  \else
    \par \addpenalty\@secpenalty\addvspace{#2}%
    \begingroup \hyphenpenalty\@M
    \@ifempty{#4}{%
      \@tempdima\csname r@tocindent\number#1\endcsname\relax
    }{%
      \@tempdima#4\relax
    }%
    \parindent\z@ \leftskip#3\relax \advance\leftskip\@tempdima\relax
    \rightskip\@pnumwidth plus4em \parfillskip-\@pnumwidth
    #5\leavevmode\hskip-\@tempdima
      \ifcase #1
       \or\or \hskip 1em \or \hskip 2em \else \hskip 3em \fi%
      #6\nobreak\relax
      \dotfill
      \hbox to\@pnumwidth{\@tocpagenum{#7}}
    \par
    \nobreak
    \endgroup
  \fi}
\makeatother

\newtheorem{theorem}{Theorem}
\numberwithin{theorem}{section}
\newtheorem{lemma}[theorem]{Lemma}
\newtheorem{proposition}[theorem]{Proposition}

\theoremstyle{definition}

\theoremstyle{remark}
\newtheorem{remark}[theorem]{Remark}

\newcommand{\R}{\mathbb{R}}

\newcommand{\Dtwo}{\mathcal{D}_k}

\newcommand{\bx}{\boldsymbol{x}}
\newcommand{\bw}{\bx'}
\newcommand{\vw}{\bx}

\newcommand{\bu}{\boldsymbol{u}}
\newcommand{\bv}{\boldsymbol{v}}
\newcommand{\bt}{\boldsymbol{t}'}
\newcommand{\be}{\begin{eqnarray*}}
\newcommand{\ee}{\end{eqnarray*}}
\newcommand{\bs}{\boldsymbol}

\newcommand{\Smone}{\mathbb{S}^{m-1}}
\newcommand{\Bm}{\mathbb{B}^{m}}
\newcommand{\BUH}{\mathbb{R}^{m-1}}

\newcommand{\vx}{\vw}
\newcommand{\vt}{\bs{t}}

\newcommand{\HK}{\mathcal{H}_k}

\newcommand{\z}{\bx}
\newcommand{\bz}{\bs{z}}

\newcommand{\ze}{\bs{\zeta}}
\newcommand{\et}{\bs{\eta}}
\newcommand{\Clm}{\mathcal{C}l_m}
\newcommand{\bnu}{\bs{\nu}}
\newcommand{\bom}{\bs{\omega}}

\numberwithin{equation}{section}

\begin{document}

\title[Boundary Value Problems for Bosonic Laplacians]{Boundary value problems in Euclidean space for Bosonic Laplacians}
\author[C. Ding]{Chao Ding}
\address{Department of Mathematics and Statistics, Masaryk University, Brno, Czech Republic}
\email{chaoding@math.muni.cz}
\author[P.T. Nguyen]{Phuoc-Tai Nguyen}
\address{Department of Mathematics and Statistics, Masaryk University, Brno, Czech Republic}
\email{ptnguyen@math.muni.cz}
\author[J. Ryan]{John Ryan}
\address{Department of Mathematical Science, University of Arkansas, Fayetteville, AR. U.S.A.}
\email{jryan@uark.edu}
\date{\today}
\begin{abstract} 
A bosonic Laplacian is a conformally invariant second order differential operator acting on smooth functions defined on domains in Euclidean space and taking values in higher order irreducible representations of the special orthogonal group. In this paper, we study boundary value problems involving bosonic Laplacians in the upper-half space and the unit ball. Poisson kernels in the upper-half space and the unit ball are constructed, which give us solutions to the Dirichlet problems with $L^p$ boundary data, $1 \leq p \leq \infty$. We also prove the uniqueness for solutions to the Dirichlet problems with continuous data for bosonic Laplacians and provide analogs of some properties of harmonic functions for null solutions of bosonic Laplacians, for instance, Cauchy's estimates, the mean-value property, Liouville's Theorem, etc.
\medskip
	
	
\medskip
	

\end{abstract}

\maketitle


\section{Introduction}
Boundary value problems are extremely important as they model a large amount of phenomena and applications, such as solid mechanics, heat transfer, fluid mechanics, acoustic diffusion, etc. Among the earliest boundary value problems to be investigated are the Dirichlet problems of Laplace's equation and Poisson's equation. These problems are frequently studied in many branches of physics, for instance, electrostatics, gravitation and fluid dynamics. Thus, the study of boundary value problems of Laplace's equation is considered as one of the most important topics in classical harmonic analysis, more details can be found in \cite{Axler,Stein}. Here we investigate a type of second order conformally invariant differential operators, named as bosonic Laplacians. These differential operators act on functions taking values in irreducible representations of $SO(m)$, hence of the spin group $Spin(m)$. In this case, these representation spaces are realized as the spaces of scalar-valued homogeneous harmonic polynomials. In particular, when the degree of the space of homogeneous harmonic polynomials is zero, the bosonic Laplacian reduces to the classical Laplacian. Bargmann and Wigner \cite{Bargmann} showed that particles should correspond to irreducible representations of the Lorentz group, labelled by a quantum number called spin. Thus, we also named the representation spaces  mentioned above as \emph{higher spin spaces}, see \cite{DeBie,DWR,Eelbode}. We discover that  bosonic Laplacians also have many properties similar to those of the classical Laplacian, for instance, Cauchy's estimates, the mean-value property, Liouville's Theorem, the Poisson kernel, etc. In order to facilitate calculations, Clifford analysis and algebras are needed.
\par
The first investigation on this type of differential operators can be traced back to the Stein-Weiss gradients introduced in \cite{Stein} in $1968$. The authors provided a technique to construct first order conformally invariant differential operators through a certain projection. Explicit expressions of conformally invariant differential operators in the higher spin spaces have been provided in \cite{Bures,Dunkl,DeBie,Eelbode,DWR} with different approaches via Clifford analysis. In particular, the second order conformally invariant differential operators, which have integer spin, are named as \emph{bosonic Laplacians} (also called the higher spin Laplace operators in \cite{DeBie}) in analogy with bosons in physics, which are particles of integer spin. 
Further, Clerc and \O rsted \cite{CO} used a representation-theoretic framework to show the relations between these conformally invariant differential operators and Knapp-Stein intertwining operators. In \cite{Ding,DR}, a Borel-Pompeiu formula and a Cauchy's (or Green's) integral formula for these conformally invariant differential operators  were provided for the first time. These results are of interest and motivate a study of boundary value problems involving the aforementioned operators which have not been well understood yet.
\par
In this paper, we carry out an investigation of boundary value problems involving bosonic Laplacians. The intricate form of these operators, together with the interaction of the two variables and the rotation action on the second variable, leads to the invalidity of some classical techniques and highly complicates the analysis. 
\par
\textbf{Main results:} The contribution of the paper is the construction of the Poisson kernels in the higher spin spaces in the upper-half space and the unit ball. This leads to the existence of solutions to the Dirichlet problems for bosonic Laplacians with $L^p$ boundary data, $1 \leq p \leq \infty$. This paper also shows a counterpart of important results regarding the Laplace operator such as the mean-value property, Cauchy's estimates and Liouville's Theorem, etc. Our results extend the results regarding the first order higher spin differential operators described by Stein and Weiss to the second order ones, and form the basis for further study of PDEs involving bosonic Laplacians, such as bosonic Hardy spaces and bosonic Bergman spaces.

\subsection*{Acknowledgements}
Chao Ding and Phuoc-Tai Nguyen are supported by Czech Science Foundation, project GJ19-14413Y.
\section{Preliminaries}
We begin with some basics of Clifford algebras. Let $\{\bs{e}_1,\cdots,\bs{e}_m\}$ be a standard orthonormal basis for the $m$-dimensional Euclidean space $\R^m$. The (real) Clifford algebra $\mathcal{C}l_m$ is generated by $\R^m$ with the relationship $\bs{e}_i\bs{e}_j+\bs{e}_j\bs{e}_i=-2\delta_{ij},\ 1\leq i,j\leq m.$
An arbitrary element of the basis of the Clifford algebra can be written as $\bs{e}_A=\bs{e}_{j_1}\cdots \bs{e}_{j_r},$ where $A=\{j_1, \cdots, j_r\}\subset \{1, 2, \cdots, m\}$ and $1\leq j_1< j_2 < \cdots < j_r \leq m$. Hence for any $a\in \mathcal{C}l_m$, we have $a=\sum_Aa_A\bs{e}_A,$ where $a_A\in \mathbb{R}$, in particular, $\bs{e}_{\emptyset}=1$. The $m$-dimensional Euclidean space $\R^m$ is embedded into $\mathcal{C}l_m$ with the mapping
$\bx=(x_1,\cdots,x_m)\ \mapsto\quad \sum_{j=1}^mx_j\bs{e}_j.$
For $\bx\in\R^m$, one can easily obtain that $|\bx|^2=\sum_{j=1}^mx_j^2=-\bx^2$. 
For $a=\sum_Aa_A\bs{e}_A\in\Clm$, we define the reversion of $a$ as
$
\widetilde{a}=\sum_{A}(-1)^{|A|(|A|-1)/2}a_A\bs{e}_A,
$
where $|A|$ is the cardinality of $A$. In particular, $\widetilde{\bs{e}_{j_1}\cdots \bs{e}_{j_r}}=\bs{e}_{j_r}\cdots \bs{e}_{j_1}$. Also $\widetilde{ab}=\widetilde{b}\widetilde{a}$ for $a, b\in\Clm$. 
\par
For $m\geq 3$, a theorem of Liouville tells us that the only conformal transformations are M\"obius transformations. Ahlfors and Vahlen showed that any M\"{o}bius transformation in $\mathbb{R}^m \cup \{\infty\}$ can be expressed as $y=(a\bx+b)(c\bx+d)^{-1}$ with $a,\ b,\ c,\ d\in \mathcal{C}l_m$ satisfying certain conditions. Since $\bs{y}=(a\bx+b)(c\bx+d)^{-1}=ac^{-1}+(b-ac^{-1}d)(c\bx+d)^{-1}$, a conformal transformation can be decomposed as compositions of translation, dilation, reflection and inversion. This gives an \emph{Iwasawa decomposition} for M\"obius transformations. See \cite{Porteous} for more details.
\par
The classical Dirac operator is defined as $D_{\bx}=\sum_{j=1}^m\partial_{x_j}\bs{e}_j$, which factorizes the Laplace operator $\Delta_{\bx}=-D_{\bx}^2$. Let $\HK$ ($1 \leq k \in {\mathbb N}$) be the space of real-valued homogeneous harmonic polynomials of degree $k$ in $m$-dimensional Euclidean space. If we consider a function $f(\bx,\bu)\in C^{\infty}(\R^m\times\R^m,\HK)$, i.e., for a fixed $\bx\in\R^m$, $f(\bx,\bu)\in\HK$ with respect to $\bu \in \R^m$. Recall that bosonic Laplacians \cite{Eelbode} are defined as 
\begin{eqnarray}\label{Dtwo}
&&\Dtwo:\ C^{\infty}(\R^m\times\R^m,\HK)\longrightarrow C^{\infty}(\R^m\times\R^m,\HK),\nonumber\\
&&\Dtwo=\Delta_{\bx}-\frac{4\langle \bu,D_{\bx}\rangle\langle D_{\bu},D_{\bx}\rangle}{m+2k-2}+\frac{4|\bu|^2\langle D_{\bu},D_{\bx}\rangle^2}{(m+2k-2)(m+2k-4)},
\end{eqnarray}
where $\langle\ ,\ \rangle$ is the standard inner product in $\R^m$. In particular, $
\mathcal{D}_1=\Delta_{\bx}-\frac{4}{m}\langle \bu,D_{\bx}\rangle\langle D_{\bu},D_{\bx}\rangle
$
 is the generalized Maxwell operator. Further, it reduces to the source-free classical Maxwell equations given in terms of the Faraday-tensor when $m=4,\ k=1$ with signature $(-,+,+,+)$. More details can be found in \cite{Eelbode}.
\section{Dirichlet problems of Bosonic Laplacians}
In this section, we investigate Dirichlet problems involving bosonic Laplacians in the upper-half space and the unit ball with different boundary data.
\subsection{Poisson kernel in the upper-half space}
We start this subsection by introducing some technical lemmas for the real-valued homogeneous harmonic polynomials. For $\bx \in \R^{m}_+$, we write $\bx=(\bx',y)$ with $\bx'=(x_1,\cdots,x_{m-1}) \in\R^{m-1}$ and $y>0$. We claim that
\begin{lemma}\label{lemma1}
Suppose that $f_k(\bu)\in\HK$ and $\bu\in\R^m$. Then there  holds
\be
c_{m,k}\int_{\R^{m-1}}\displaystyle\frac{y}{|\bx|^m}f_k\bigg(\displaystyle\frac{\bx\bu\bx}{|\bx|^2}\bigg)d\bx'=f_k(\bu),
\ee
where $c_{m,k}=\frac{2(m+2k-2)}{(m-2)\omega_m}$ and $\omega_m$ is the surface area of the unit sphere $\Smone$.
\end{lemma}
\begin{proof}
Since $f_k(\bu)\in\HK$, we see that $f_k\big(\frac{\bx\bu\bx}{|\bx|^2}\big)$ is bounded for each fixed $\bu\in\R^m$, $\bx'\in\R^{m-1}$, and
$\frac{y}{|\bx|^m} \in L^1(\R^{m-1},d\bx')$ (see  \cite[Chapter 7]{Axler}). Further, if we let $\bx=y\bs{z}$, we derive that the integral $\int_{\R^{m-1}}\frac{y}{|\bx|^m}f_k\big(\frac{\bx\bu\bx}{|\bx|^2}\big)d\bx'$ does not depend on the value of $y$. Hence, we can assume that
$$\displaystyle\int_{\R^{m-1}}\displaystyle\frac{y}{|\bx|^m}f_k\bigg(\displaystyle\frac{\bx\bu\bx}{|\bx|^2}\bigg)d\bx'=\frac{P(\bu)}{c_{m,k}},$$
where $P(\bu)$ is a function of $\bu$ and $c_{m,k}$ is a non-zero constant only depending on $m,k$, and it will be determined later on. On the one hand, we have
\begin{eqnarray}\label{eqn1}
\int_{0}^{\infty}\frac{1}{1+y^2}dy\displaystyle\int_{\R^{m-1}}\displaystyle\frac{y}{|\bx|^m}f_k\bigg(\displaystyle\frac{\bx\bu\bx}{|\bx|^2}\bigg)d\bx'
=\frac{\pi }{2c_{m,k}}P(\bu).
\end{eqnarray}
On the other hand, if we denote $\bx=r\bs{\zeta}$, where $\bs{\zeta}=(\zeta_1,\cdots,\zeta_m)\in\Smone$, then we obtain
\begin{align}\label{eqn2}
&\int_{0}^{\infty}\frac{1}{1+y^2}dy\displaystyle\int_{\R^{m-1}}\displaystyle\frac{y}{|\bx|^m}f_k\bigg(\displaystyle\frac{\bx\bu\bx}{|\bx|^2}\bigg)d\bx'
=\int_{\R^m_+}\frac{y}{1+y^2}|\bx|^{-m}f_k\bigg(\displaystyle\frac{\vw\bu\vw}{|\bx|^2}\bigg)d\bx\nonumber\\
=&\int_{\Smone_+}\int_0^{\infty}\frac{\zeta_m}{1+(r\zeta_m)^2}drf_k(\bs{\zeta}\bu\bs{\zeta})dS(\bs{\zeta})=\frac{\pi}{2}\int_{\Smone_+}f_k(\bs{\zeta}\bu\bs{\zeta})dS(\bs{\zeta}).
\end{align} 
Further, we notice that $\bs{\zeta}\bu\bs{\zeta}$ is invariant under $\bs{\zeta}\rightarrow -\bs{\zeta}$. Therefore, one obtains
\be
\int_{\Smone_+}f_k(\bs{\zeta}\bu\bs{\zeta})dS(\bs{\zeta})=\frac{1}{2}\int_{\Smone}f_k(\bs{\zeta}\bu\bs{\zeta})dS(\bs{\zeta}).
\ee 
From \eqref{eqn1},\eqref{eqn2} and \cite[Lemma 6]{Dunkl}, we observe that
\be
\frac{\pi }{2c_{m,k}}P(\bu)=\frac{\pi}{4}\cdot\frac{(m-2)\omega_m}{m+2k-2}  f_k(\bu).
\ee
Therefore, we have $P(\bu)=f_k(\bu)$ and $c_{m,k}=\frac{2(m+2k-2)}{(m-2)\omega_m} $ as desired.
\end{proof}
Let $Z_k(\bu,\bv)$ be the reproducing kernel of the spherical harmonics (see \cite{Axler}) in the sense that 
\be
f(\bv)=\int_{\mathbb{S}^{m-1}}Z_k(\bu,\bv)f(\bu)dS(\bu),\ \text{for all}\ f(\bv)\in\mathcal{H}_k.
\ee
Notice that a bosonic Laplacian $\Dtwo$ given in (\ref{Dtwo}) is a second-order differential operator with respect to $\bx$, then $y Z_k(\bu,\bv)$ is a trivial null solution of $\Dtwo$. Further, in \cite{DeBie,DWR}, the authors showed that $\Dtwo$ is a second-order conformally invariant differential operator, which is particularly conformally invariant under the following special conformal transformation $K:\ f(\vw,\bu)\mapsto K[f](\vw,\bu):=|\bx|^{2-m}f\big(\frac{\vw}{|\bx|^2},\frac{\vw\bu\vw}{|\bx|^2}\big).$ We apply $K$ to $yZ_k(\bu,\bv)$ to obtain a non-trivial null solution of $\Dtwo$ given by
\be
K[yZ_k(\bu,\bv)]=|\bx|^{2-m}\frac{y}{|\bx|^2}Z_k\bigg(\frac{\vw\bu\vw}{|\bx|^2},\bv\bigg)=\frac{y}{|\bx|^{m}}Z_k\bigg(\frac{\vw\bu\vw}{|\bx|^2},\bv\bigg).
\ee
\par
Now, let $\bt\in\R^{m-1},\ \vt=(\bt,0)=(t_1,\cdots,t_{m-1},0),\ \bx\in\R^m_+,\ \bu\in\Smone,\ \bv\in\Bm$, and set
\begin{eqnarray}\label{Pkernel}
P_H(\bx,\vt,\bu,\bv)=c_{m,k}\frac{y}{|\bx-\vt|^m}Z_k\bigg(\frac{(\vw-\vt)\bu(\vw-\vt)}{|\bx-\vt|^2},\bv\bigg).
\end{eqnarray}
The function $P_H$ is called the \emph{Poisson kernel} of bosonic Laplacians in the upper-half space. Notice that, in accordance to Lemma \ref{lemma1}, we have
\be
\int_{\R^{m-1}}P_H(\bx,\vt,\bu,\bv)d\bt=Z_k(\bu,\bv) \quad \bx\in\R^m_+,\ \bu\in\Smone,\ \bv\in\Bm.
\ee
Now, we introduce an approximation property of the Poisson kernel $P_H$.
\begin{proposition}\label{prop0}
Let $\bs{a}\in\R^{m-1}$, $\vw\in\R^m_+$, $\bv\in\Bm$ and $\delta>0$, then we have
\begin{equation} \label{limPH-1}
\lim_{\vx \to \bs{a}}\int_{|\bt-\bs{a}|>\delta}\int_{\Smone}P_H(\vw,\vt,\bu,\bv)dS(\bu)d\bt = 0.
\end{equation}
\end{proposition}
\begin{proof}
Notice that
\be
&&\bigg{|}\int_{|\bt-\bs{a}|>\delta}\int_{\Smone}P_H(\vw,\vt,\bu,\bv)dS(\bu)d\bt\bigg{|}\\
&\leq&\int_{|\bt-\bs{a}|>\delta}\int_{\Smone}\bigg{|}c_{m,k}\frac{y}{|\bx-\vt|^m}Z_k\bigg(\frac{(\vw-\vt)\bu(\vw-\vt)}{|\bx-\vt|^2},\bv\bigg)\bigg{|}dS(\bu)d\bt\\
&\leq&c_{m,k}\omega_m\dim\HK\int_{|\bt-\bs{a}|>\delta}\frac{y}{|\bx-\vt|^m}d\bt,
\ee
where we have used the estimate $|Z_k(\bu,\bv)|\leq \dim\HK$ for $\bu,\bv\in\Smone$ (see  \cite[Proposition 5.27]{Axler}). By \cite[Lemma 1.3.5 (c)]{Ban}, we obtain \eqref{limPH-1}.
\end{proof}~\\
 For $1\leq p<\infty$, let $L^p(\R^{m-1}\times\Bm,\HK)$ be the space of Borel measurable functions $f$ on $\R^{m-1}\times\Bm$ for which
\begin{eqnarray}\label{norm}
\|f\|_{L^p(\R^{m-1}\times\Bm,\HK)}=\bigg(\int_{\R^{m-1}}\int_{\Smone}|f(\bw,\bu)|^pdS(\bu)d\bx'\bigg)^{1/p}<+\infty.
\end{eqnarray}
$L^{\infty}(\R^{m-1}\times\Bm,\HK)$ consists of the Borel measurable functions $f$ on $\R^{m-1}\times\Bm$ for which $\|f\|_{\infty}<+\infty$, where $\|\cdot \|_{\infty}$ stands for the essential supremum norm on $\R^{m-1}\times\Bm$. One might notice that, when we define the norm of $L^p$, the integration of $\bu$ is over $\Smone$ instead of $\Bm$, although $\bu\in\Bm$. This is because $f(\bw,\bu)$ is a homogeneous harmonic polynomial of degree $k$ in $\bu$, so the norm defined on $\Smone$ is the same as the norm defined in $\Bm$ up to a multiplicative constant. Indeed, one can easily see that 
\begin{equation} \label{equivalence} \|f\|^p_{L^p(\R^{m-1}\times\Bm,\HK)}=(m+kp)^{-1}\|f\|^p_{L^p(\R^{m-1}\times\Smone,\HK)}.
\end{equation}
The \emph{Poisson integral} of $f\in L^p(\R^{m-1}\times\Bm,\HK)$, $p\in[1,\infty]$, is given by
\begin{equation} \label{Poisson-1}
P_H[f](\vw,\bv)=\int_{\R^{m-1}}\int_{\Smone}P_H(\vw,\vt,\bu,\bv)f(\bt,\bu)dS(\bu)d\bt,
\end{equation}
where $\bx\in\R^m_+,\ \bv\in\Bm$. By the boundedness of $Z_k(\bu,\bv)$, the fact that $\frac{y}{|\bx-\vt|^m}\in L^q(\R^{m-1}),\ 1/p+1/q=1$, and H\"{o}lder's inequality, we observe that $P_H[f]$ is well-defined for every $\vw\in\R^m_+,\ \bv\in\Bm$.
\subsection{Dirichlet problem in the upper-half space}
Now, we claim that the Poisson integral given in \eqref{Poisson-1} solves the following Dirichlet problem in $\R^m_+$.
\begin{theorem}[Dirichlet problem in $\R^{m}_+$ with continuous and bounded data]\label{theorem1}\par
Suppose $f \in C(\R^{m-1}\times\Bm, \HK) \cap L^\infty(\R^{m-1}\times\Bm, \HK)$. Define $g$ on $\overline{\R^{m}_+}\times\Bm$ by
\[
g(\vx,\bv)=\begin{cases}
P_H[f](\vx,\bv),\ &\text{if}\ \vx\in\R^m_+, \quad \bv \in \Bm,\\
f(\bw,\bv),\ &\text{if}\ \bw\in\R^{m-1}, \quad \bv \in \Bm.
\end{cases}
\]
Then $g$ is continuous in $\overline{\R^m_+}$ with respect to $\vw$, $\Dtwo g=0$ in $\R^m_+\times\Bm$ and 
\begin{equation} \label{g<f-1} \|g\|_{L^{\infty}(\R^m_+\times\Bm)}\leq a_{m,k}\|f\|_{L^{\infty}(\R^{m-1}\times\Bm)},
\end{equation} 
where $a_{m,k}>0$ is a constant only depending on $m$ and $k$.
\end{theorem}
\begin{proof}
First, from the expression of $P_H(\vx,\vt,\bu,\bv)$ given in (\ref{Pkernel}), if we consider a fixed derivative of $P_H(\vx,\vt,\bu,\bv)$ with respect to $y,\ x_j,\ j=1,\cdots,m-1$ or $u_j,\ j=1,\cdots, m$, and assume $V$ is a compact subset in $\R^m_+$, then for any $(\bw,y)\in V$ and $\bu\in\Smone$, the supremum of the derivative as a function of $\vt$ and $\bv$ , is integrable. This is due to the fact that $P_H(\vx,\vt,\bu,\bv)$ is nonsingular when $\vx$ is in a compact set in $\R^m_+$. Hence, we can differentiate under the integral sign to obtain that $\Dtwo P_H[f]=0$. 

Next we will show that $g(\vx,\bv)$ is continuous in $\R^m_+$ with respect to $\vx$. Let $\bs{a} \in \R^{m-1}$ and $\bv \in \Bm$. By Lemma \ref{lemma1}, we have 
\begin{equation} \label{fav} \begin{aligned}
f(\bs{a},\bv)=&\int_{\Smone}Z_k(\bu,\bv)f(\bs{a},\bu)dS(\bu)\\
=&\int_{\R^{m-1}}\int_{\Smone}P_H(\vx,\vt,\bu,\bv)f(\bs{a},\bu)dS(\bu)d\bt.
\end{aligned} \end{equation}
Let $\delta>0$. From \eqref{fav} and the estimate $|Z_k(\bu,\bv)|\leq\dim\HK$, we obtain 
\begin{align*} 
&|g(\vx,\bv)-f(\bs{a},\bv)|=|P_H[f](\vx,\bv)-f(\bs{a},\bv)|\\
&\leq \int_{\R^{m-1}}\int_{\Smone}|P_H(\vx,\vt,\bu,\bv)(f(\bt,\bu)-f(\bs{a},\bu))|dS(\bu)d\bt \\
&\leq c_{m,k}\dim\HK\cdot\int_{|\bt-\bs{a}|<\delta}\int_{\Smone}\frac{y}{|\bx-\vt|^m}\cdot |f(\bt,\bu)-f(\bs{a},\bu)|dS(\bu)d\bt\\
&+c_{m,k}\dim\HK\cdot\int_{|\bt-\bs{a}|>\delta}\int_{\Smone}\frac{y}{|\bx-\vt|^m}\cdot |f(\bt,\bu)-f(\bs{a},\bu)|dS(\bu)d\bt.
\end{align*}
Take arbitrarily $\epsilon>0$, we can choose $\delta>0$ small enough such that the first  integral is smaller than $\epsilon$ due to the facts that $f(\bt,\bu)$ is continuous at $\bs{a}$ and $\int_{\R^{m-1}}\frac{y}{|\bx-\vt|^m}d\bt$ is bounded. The second integral above approaching zero when $\vx\rightarrow \bs{a}$ can be immediately obtained from Proposition \ref{prop0}. This completes the proof of the continuity.
\par
Finally, we will prove \eqref{g<f-1}. For any $\vx\in\R^m_+$ and $\bv\in\Bm$, we have
\begin{align*}
&|g(\vx,\bv)|=\bigg{|} \int_{\R^{m-1}}\int_{\Smone}P_H(\vx,\vt,\bu,\bv)f(\bt,\bu)dS(\bu)d\bt\bigg{|}\\
\leq& c_{m,k}\int_{\R^{m-1}}\frac{y}{|\bx-\vt|^m}\int_{\Smone}\bigg{|}Z_k\bigg(\displaystyle\frac{(\vx-\vt)\bu(\vx-\vt)}{|\bx-\vt|^2},\bv\bigg)f(\bt,\bu)\bigg{|}dS(\bu)d\bt\\
\leq& c_{m,k}\dim\HK\cdot\omega_m\|f\|_{L^{\infty}(\R^{m-1}\times\Bm)}\int_{\R^{m-1}}\frac{y}{|\bx-\vt|^m}d\bt\\
=&:a_{m,k}\|f\|_{L^{\infty}(\R^{m-1}\times\Bm)}.
\end{align*}
The last second equality comes from \cite[Lemma 1.3.5]{Ban}. Hence, we obtain \eqref{g<f-1}.
\end{proof}
We now extend the result to the Dirichlet problem with $L^p$ data.
\begin{theorem}[Dirichlet problem with $L^p$ data]
Assume that $1\leq p< \infty$ and $f\in L^p(\R^{m-1}\times\Bm,\HK)$. Let $g(\vx,\bv)=P_H[f](\vx,\bv)=P_H[f](\bw,y,\bv)$, and for $y>0$, we set $g_y(\bw,\bv)=P_H[f](\bw,y,\bv)$. 
\begin{enumerate}
\item There holds $\Dtwo g=0$ in $\R^m_+\times\Bm$.
\item  There exists a constant $c'_{m,k}>0$ depending on $m$ and $k$ such that, for any $y>0$, we obtain $\|g_y\|_{L^p(\R^{m-1}\times\Bm,\HK)}\leq c'_{m,k}\|f\|_{L^p(\R^{m-1}\times\Bm,\HK)}.$
\item We have  $\|g_y-f\|_{L^p(\R^{m-1}\times\Bm,\HK)} \to 0$ as $y \to 0$.
\end{enumerate}
\end{theorem}
\begin{proof}
For a fixed $\vt=(\bt,0)$ with $\bt \in\R^{m-1}$, we have $\Dtwo P_H(\vx,\vt,\bu,\bv)=0$. Further, from the expression that $P_H(\vx,\vt,\bu,\bv)$  in (\ref{Pkernel}), if we consider a fixed derivative of $P_H(\vx,\vt,\bu,\bv)$ with respect to $y,\ x_j,\ j=1,\cdots,m-1$ or $u_j,\ j=1,\cdots, m$, and assume $V$ is a compact subset in $\R^m_+$, then for any $(\bw,y)\in V$ and $\bu\in\Smone$, the supremum of the derivative as a function of $\vt$ and $\bv$ , is in $L^q(\R^{m-1}\times\Bm,\HK)$, where $1/p+1/q=1$. This is because $P_H(\vx,\vt,\bu,\bv)$ is nonsingular when $\vx$ is in a compact set in $\R^m_+$. This allows us to differentiate under the integral sign to conclude that $\Dtwo P_H[f]=0$. 

Next, we prove statement (2). By the Minkowski integral inequality, we have
\begin{align*}
&\|g_y\|_{L^p(\R^{m-1}\times\Bm,\HK)}\\
=&c_{m,k}\bigg(\int_{\R^{m-1}}\int_{\Smone}\bigg{|}\int_{\R^{m-1}}\int_{\Smone}P_H(\vx,\vt,\bu,\bv)f(\bt,\bu)dS(\bu)d\bt\bigg{|}^pdS(\bv)d\bx'\bigg)^{\frac{1}{p}}\\
\leq& c'_{m,k}\int_{\R^{m-1}}\int_{\Smone}\bigg(\int_{\R^{m-1}}\int_{\Smone}\bigg{|}\displaystyle\frac{y}{|\bs{\eta}|^m}f(\vx-\bs{\eta},\bu)\bigg{|}^pdS(\bv)d\bx'\bigg)^{\frac{1}{p}}dS(\bu)d{\bs{\eta}}'\\
\leq& c_{m,k}''\int_{\R^{m-1}}\frac{y}{|\bs{\eta}|^m}d{\bs{\eta}}'\bigg(\int_{\Smone}\int_{\R^{m-1}}|f(\vx-\bs{\eta},\bu)|^pd\bx'dS(\bu)\bigg)^{\frac{1}{p}},
\end{align*}
where $\bs{\eta}=(\bx'-\bt,y)=:(\bs{\eta}',y)$. This yields statement (2).
\par

To prove statement (3), we denote $f_{\bs{\eta}}(\vx,\bu)=f(\vx-\bs{\eta},\bu)$. Let $\epsilon>0$ and in accordance to the continuity of the translation operator in the $L^p$ norm with $\ p<\infty$, we can choose $\delta>0$ so that $\|f_{\bs{\eta}}-f\|_{L^p(\R^{m-1}\times\Bm,\HK)}\leq \epsilon$, when $|\bs{\eta}|<\delta$. Then by the Minkowski integral inequality and the estimate $|P_H(\bs{\eta},0,\bu,\bv)|\leq \dim\HK\displaystyle\frac{y}{|\bs{\eta}|^m}$, we obtain
\begin{align*}
&\|g_y-f\|_{L^p(\R^{m-1}\times\Bm,\HK)}\\
&\leq\int_{\R^{m-1}}\int_{\Smone}\left(\int_{\R^{m-1}}\int_{\Smone}|P_H(\bs{\eta},0,\bu,\bv)(f_{\bs{\eta}}(\vx,\bu)-f(\bx',\bu))|^pdS(\bv)d\bx'\right)^{\frac{1}{p}} \\
&\hspace{11.5cm}   dS(\bu)d{\bs{\eta}}' \\
&\leq a_{m,k}'\int_{\R^{m-1}}\frac{y}{|\bs{\eta}|^m}\int_{\Smone}\left(\int_{\R^{m-1}}|f_{\bs{\eta}}(\vx,\bu)-f(\bw,\bu)|^pd\bx'\right)^{\frac{1}{p}}dS(\bu)d{\bs{\eta}}'\\
&=a_{m,k}'\int_{\R^{m-1}}\frac{y}{|\bs{\eta}|^m}\|f_{\bs{\eta}}-f\|_{L^p(\R^{m-1}\times\Bm,\HK)}d{\bs{\eta}}'.
\end{align*}
This implies
\begin{align*}
&\|g_y-f\|_{L^p(\R^{m-1}\times\Bm,\HK)}\\
=&a_{m,k}'\bigg(\int_{|\bs{\eta}|<\delta}\frac{y}{|\bs{\eta}|^m}\|f_{\bs{\eta}}-f\|_{L^p(\R^{m-1}\times\Bm,\HK)}d{\bs{\eta}}'\\
&\quad\quad+\int_{|\bs{\eta}|>\delta}\frac{y}{|\bs{\eta}|^m}\|f_{\bs{\eta}}-f\|_{L^p(\R^{m-1}\times\Bm,\HK)}d{\bs{\eta}}'\bigg)\\
\leq&a_{m,k}'\bigg(\epsilon\int_{|\bs{\eta}|<\delta}\frac{y}{|\bs{\eta}|^m}d{\bs{\eta}}'+2\|f\|_{L^p(\R^{m-1}\times\Bm,\HK)}\int_{|\bs{\eta}|>\delta}\frac{y}{|\bs{\eta}|^m}d{\bs{\eta}}'\bigg).
\end{align*}
By a similar argument as in Theorem \ref{theorem1}, we derive that the above two integrals approach zero when $y$ goes to zero. This completes the proof of the last statement.
\end{proof}
\subsection{Dirichlet problem in the unit ball}
In this subsection, we will derive the Poisson kernel for Bosonic Laplacians in the unit ball by applying a particular Cayley transform, which is a M\"obius transformation mapping the unit ball to the upper-half space. We need the following technical lemma.
\begin{lemma}\label{lem1}
Let $\bx,\ze\in\R^m$ and $\phi(\bx)=(a\bx+b)(c\bx+d)^{-1}$ be a M\"{o}bius transformation. Then, one has
\begin{enumerate}
\item $|\phi(\bx)-\phi(\ze)|=|c\bx+d|^{-1}|\bx-\ze||c\ze+d|^{-1},$
\item $Z_k\bigg(\displaystyle\frac{(\phi(\bx)-\phi(\ze))\bu(\phi(\bx)-\phi(\ze))}{|\phi(\bx)-\phi(\ze)|^2},\bv\bigg)=Z_k\bigg(\displaystyle\frac{(\bx-\ze)\bs{\omega}(\bx-\ze)}{|\bx-\ze|^2},\bs{\nu}\bigg),$
\end{enumerate}
where
\begin{eqnarray*}
\bs{\omega}=\displaystyle\frac{\widetilde{(c\ze+d)}\bu(c\ze+d)}{|c\ze+d|^2}\quad \text{and}\quad \bs{\nu}=\displaystyle\frac{\widetilde{(c\bx+d)}\bv(c\bx+d)}{|c\bx+d|^2}.
\end{eqnarray*}
\end{lemma}
\begin{proof}
The proof is similar to that of \cite[Theorem 8]{Dunkl} and \cite[Theorem 5.1]{JR}. We briefly recall it here. The idea is to prove the identities are true for a translation, a dilation, a rotation and an inversion, separately. Then, it is also true for an arbitrary M\"{o}bius transformation in accordance to the Iwasawa decomposition. It is easy to check that $(1)$ and $(2)$ are true under a dilation, a translation or a reflection. Here, we only show details for the inversion case, i.e., $\phi(\bx)=\bx^{-1}$, i.e., $a=d=0,b=c=1$. We see that
\be
\bx^{-1}-\ze^{-1}=-\ze^{-1}(\bx-\ze)\bx^{-1}=-\bx^{-1}(\bx-\ze)\ze^{-1}.
\ee
One can observe that statement $(1)$ is true. For statement $(2)$, one has
\be
Z_k\bigg(\displaystyle\frac{\bx(\bx-\ze)\ze\bu\widetilde{\ze}(\widetilde{\bx-\ze})\widetilde{\bx}}{|\bx|^2|\bx-\ze|^2|\ze|^2},\bv\bigg)=Z_k\bigg(\displaystyle\frac{(\bx-\ze)\ze\bu\widetilde{\ze}(\widetilde{\bx-\ze})}{|\bx-\ze|^2|\ze|^2},\frac{\bx\bv\bx}{|\bx|^2}\bigg),
\ee
which is derived from the fact that the reproducing kernel of the spherical harmonics is invariant under reflection. This can be observed from the explicit expression of the reproducing kernel given in \cite[Theorem 5.38]{Axler}. Therefore, statement $(2)$ is also true under inversion. Hence, both statements are correct in accordance to the Iwasawa decomposition.
\end{proof}
Let $\mathbb{B}^m$ be the open unit ball in $\R^m$ and $\varphi$ is the Cayley transform given as follows.
\begin{eqnarray}\label{Cay}
\varphi: \mathbb{B}^m\longrightarrow \R^{m}_+,\quad \bx\mapsto \bs{z}=-\frac{1}{2}(\bx+\bs{e}_m)(\bs{e}_m\bx+1)^{-1}.
\end{eqnarray}
In particular, if $\ze\in\Smone$ then $\varphi(\ze)\in\BUH$. Let $\bx,\bnu\in\Bm$ and $P_B[h](\bx,\bnu)$ is defined as
\begin{eqnarray*}
\frac{c_{m,k}}{2}\int_{\Smone}\int_{\Smone}\frac{1-|\bx|^2}{|\bx-\ze|^m}Z_k\bigg[\frac{(\bx-\ze)\bs{\omega}(\bx-\ze)}{|\bx-\ze|^2},\bs{\nu}\bigg]h(\ze,\bom)dS(\bom)dS(\ze),
\end{eqnarray*}
where $c_{m,k}$ is given in Lemma \ref{lemma1}. Here, we remind the reader that $\bs{\omega}$ and $\bs{\nu}$ are actually obtained by rotations of $\bu$ and $\bv$, respectively (see more details in \cite[Theorem 6.3]{Gilbert}). We claim that
\begin{theorem}[Dirichlet problem in $\mathbb{B}^m$ with continuous data]\label{Pintegral}
\par
Suppose $h\in C(\Smone\times\Bm, \HK)$. Define $h^*$ in $\mathbb{B}^m\times\mathbb{B}^m$ by
\[
h^*(\bx,\bnu)=\begin{cases}
P_B[h](\bx,\bnu),\ &\text{if}\ \bx\in\mathbb{B}^m, \quad \bnu \in \Bm\\
h(\bx,\bnu),\ &\text{if}\ \bx\in\Smone, \quad \bnu \in \Bm.
\end{cases}
\]
Then $h^*$ is continuous in $\overline{\mathbb{B}^m}$ with respect to $\bx$, $\Dtwo h^*=0$ in $\mathbb{B}^m\times\Bm$ and \begin{equation} \label{h*<h-1}
\|h^*\|_{L^{\infty}(\mathbb{B}^m\times\mathbb{B}^m)}\leq a'_{m,k}\|h\|_{L^{\infty}(\Smone\times\mathbb{B}^m)},
\end{equation} where $a'_{m,k}$ is a positive constant only depending on $m$ and $k$.
\end{theorem}
\begin{proof} We will prove this theorem by showing that it can be derived from Theorem \ref{theorem1} after applying the Cayley transform $\varphi$. Recall that, in the upper-half space case, the function
\be
P_H[f](\bs{z},\bv)=\int_{\R^{m-1}}\int_{\Smone}P_H(\bs{z},\vt,\bu,\bv)f(\bt,\bu)dS(\bu)d\bt, \;\; {\bs{z} \in\R^m_+,\ \bv \in \Bm}, 
\ee
solves the Dirichlet problem of bosonic Laplacians, where $\bs{z}=(\bs{z}',z_m)\in\R^m_+$. By plugging $\varphi(\bx)=\bs{z}$ and $\varphi(\ze)=\vt$ as defined in (\ref{Cay}), where $\bx\in\mathbb{B}^m$ and $\ze\in\Smone$, into $P_H[f](\bs{z},\bv)$ and by taking into account that $z_m=\frac{1-|\bx|^2}{2|\bs{e}_m\bx+1|^{2}}$, we obtain
\begin{align*}
&P_H[f](\bs{z},\bv)=\int_{\R^{m-1}}\int_{\Smone}\frac{c_{m,k}z_m}{|\bz-\vt|^m}Z_k\bigg(\displaystyle\frac{(\bz-\vt)\bu(\bz-\vt)}{|\bz-\vt|^2},\bv\bigg)f(\bt,\bu)dS(\bu)d\bt\\
&=\int_{\Smone}\int_{\Smone}\frac{c_{m,k}(1-|\bx|^2)|\bs{e}_m\bx+1|^{-2}}{2|\bs{e}_m\bx+1|^{-m}|\bx-\ze|^m|\bs{e}_m\ze+1|^{-m}}Z_k\bigg(\frac{(\bx-\ze)\bs{\omega}(\bx-\ze)}{|\bx-\ze|^2},\bs{\nu}\bigg)\\
&\hspace{7.2cm} \times f(\varphi(\ze),\bu)J(\varphi,\ze)dS(\bu)dS(\ze),
\end{align*}
where $J(\varphi,\ze)$ is the Jacobian and 
\begin{eqnarray*}
\bs{\omega}=\displaystyle\frac{\widetilde{(\bs{e}_m\ze+1)}\bu(\bs{e}_m\ze+1)}{|\bs{e}_m\ze+1|^2}\quad \text{and}\quad \bs{\nu}=\displaystyle\frac{(\widetilde{\bs{e}_m\bx+1)}\bv(\bs{e}_m\bx+1)}{|\bs{e}_m\bx+1|^2}.
\end{eqnarray*}
It can be checked that $J(\varphi,\ze)=|\bs{e}_m\ze+1|^{-2m+2}$, which implies
\begin{align*}
P_H[f](\varphi&(\bx),\bv)=\frac{c_{m,k}}{2}|\bs{e}_m\bx+1|^{m-2}\int_{\Smone}\int_{\Smone}\frac{1-|\bx|^2}{|\bx-\ze|^m}\\
&\times Z_k\bigg(\frac{(\bx-\ze)\bs{\omega}(\bx-\ze)}{|\bx-\ze|^2},\bs{\nu}\bigg)|\bs{e}_m\ze+1|^{2-m}f(\varphi(\ze),\bu)dS(\bu)dS(\ze).
\end{align*}
In other words, one has
\begin{align}\label{PH1}
&|\bs{e}_m\bx+1|^{2-m}P_H[f](\varphi(\bx),\bv)=\frac{c_{m,k}}{2}\int_{\Smone}\int_{\Smone}\frac{1-|\bx|^2}{|\bx-\ze|^m}\nonumber\\
&\quad \times Z_k\bigg(\frac{(\bx-\ze)\bs{\omega}(\bx-\ze)}{|\bx-\ze|^2},\bs{\nu}\bigg)|\bs{e}_m\ze+1|^{2-m}f(\varphi(\ze),\bu)dS(\bu)dS(\ze).
\end{align} 
If we let $h(\ze,\bnu)=|\bs{e}_m\ze+1|^{2-m}f(\varphi(\ze),\bv)$, we obtain from the definition of $P_B$ that
\begin{equation}  \begin{aligned}\label{PH2}
&P_B[h](\bx,\bnu)\\
&=\frac{c_{m,k}}{2}\int_{\Smone}\int_{\Smone}\frac{1-|\bx|^2}{|\bx-\ze|^m}Z_k\bigg(\frac{(\bx-\ze)\bs{\omega}(\bx-\ze)}{|\bx-\ze|^2},\bs{\nu}\bigg)h(\ze,\bom)dS(\bom)dS(\ze).
\end{aligned} \end{equation}
From \eqref{PH1} and \eqref{PH2}, we have  $P_B[h](\bx,\bnu)=|\bs{e}_m\bx+1|^{2-m}P_H[f](\varphi(\bx),\bv)$. Since $\Dtwo P_H[f]=0$ and $\Dtwo$ is conformally invariant \cite[Theorem 1]{DR}, it follows that 
$$\Dtwo P_B[h](\bx,\bnu)=|\bs{e}_m\bx+1|^{2+m}\Dtwo |\bs{e}_m\bx+1|^{2-m}P_H[f](\varphi(\bx),\bv)=0.$$ 

With the boundedness of $Z_k(\bu,\bv)$ and noticing that $\frac{1-|\bx|^2}{|\bx-\ze|^m}$ is the Poisson kernel of the Laplacian in the unit ball, we derive \eqref{h*<h-1} by using a similar argument used in the upper-half space case. This completes the proof of Theorem \ref{Pintegral}.
\end{proof}
From Theorem \ref{Pintegral}, we can denote the Poisson kernel of bosonic Laplacian in the unit ball by
\begin{eqnarray}\label{pkernel}
P_B(\bx,\ze,\bom,\bnu)=\frac{c_{m,k}}{2}\frac{1-|\bx|^2}{|\bx-\ze|^m}Z_k\bigg(\frac{(\bx-\ze)\bs{\omega}(\bx-\ze)}{|\bx-\ze|^2},\bs{\nu}\bigg),
\end{eqnarray}
where $\bx,\bnu\in\Bm,\ze,\bom\in\Smone$. 
\par
Now, we deal with the $L^p$ data with techniques from \cite{Axler}. Here, for $1\leq p<\infty$, $L^p(\Smone\times\mathbb{B}^m,\HK)$ consists of the Borel measurable functions $f$ on $\Smone\times\mathbb{B}^m$, for which
\be
\|f\|_{L^p(\Smone\times\mathbb{B}^m,\HK)}=\bigg(\int_{\Smone}\int_{\Smone}|f(\bx,\bu)|^pdS(\bu)dS(\bx)\bigg)^{1/p}<+\infty.
\ee
$L^{\infty}(\Smone\times\Bm,\HK)$ consists of the Borel measurable functions $f$ on $\Smone\times\Bm$ for which $\|f\|_{\infty}<+\infty$, where $\|\cdot \|_{\infty}$ stands for the essential supremum norm on $\Smone\times\Bm$. We will use $\|\cdot\|_{p}$ to represent $\|\cdot\|_{L^p(\Smone\times\mathbb{B}^m,\HK)}$ in the rest of this section when there is no confusion. For $f\in L^p(\Bm\times\Bm,\HK)$, we define $f_r(\et,\bu)=f(r\et,\bu)$ with $\et\in\Smone$, $\bu \in \Bm$, and $0\leq r<1$. Now, we claim that
\begin{theorem}\label{thmbdd}
Let $1\leq p\leq\infty$, $f\in C(\Smone\times\mathbb{B}^m,\HK)$ and $g=P_B[f]$. Then, for any $0\leq r<1$, we have $\|g_r\|_p\leq b_{m,k}\|f\|_p$, where $b_{m,k}$ is a constant depending on $m$ and $k$.
\end{theorem}
\begin{proof}
We firstly prove the result for $1\leq p<\infty$. Since $Z_k(\bu,\bv)$ is bounded, we have
\be
&&\|g_r\|_{p}^p=\int_{\Smone}\int_{\Smone}|g_r(\et,\bnu)|^pdS(\bnu)dS(\et)
\\
&\leq&b_{m,k}\int_{\Smone}\bigg\rvert\int_{\Smone}\int_{\Smone}\frac{1-r^2}{|r\et-\ze|^m}|f(\ze,\bom)|dS(\bom)dS(\ze)\bigg\rvert^pdS(\et).
\ee
Recall that $\int_{\Smone}\frac{1-r^2}{|r\et-\ze|^m}dS(\ze)=\omega_m,$
then according to the above estimate and the Jensen's inequality, we have
\be 
\|g_r\|_{p}^p 
&\leq&b_{m,k}\int_{\Smone}\int_{\Smone}\int_{\Smone}\frac{1-r^2}{|r\ze-\et|^m}dS(\et)|f(\ze,\bom)|^pdS(\bom)dS(\ze)\\
&=&b_{m,k}\omega_m\int_{\Smone}\int_{\Smone}|f(\ze,\bom)|^pdS(\bom)dS(\ze)=b_{m,k}\|f\|_p.
\ee
We used the fact that $|r\ze-\et|=|r\et-\ze|,\ \ze,\et\in\Smone$ and Fubini's Theorem in the above calculation. 

For $p=\infty$, it is easy to observe that
\begin{align*}
|g_r(\et,\bnu)|&=\bigg\rvert\int_{\Smone}\int_{\Smone}P_B(r\et,\ze,\bom,\bnu)f(\ze,\bom)dS(\bom)dS(\ze)\bigg\rvert\\
&\leq b'_{m,k}\|f\|_{\infty}\int_{\Smone}\frac{1-r^2}{|r\et-\ze|^m}dS(\ze)=b_{m,k}'\|f\|_{\infty},
\end{align*}
which completes the proof of the theorem.
\end{proof}
An immediate consequence of the theorem above is the following.
\begin{proposition} Let $1 \leq p \leq \infty$. Assume that $\Dtwo f=0$ in $\Bm\times\Bm$ and $0\leq r\leq s<1$. Then $\|f_r\|_p\leq b_{m,k}\|f_s\|_p$.
\end{proposition}
\begin{proof} From Theorem \ref{thmbdd}, we have
$\|f_r\|_p=\|P_B[f_s]_{\frac{r}{s}}\|_p\leq b_{m,k}\|f_s\|_p$.
\end{proof}
Theorem \ref{thmbdd} also implies a $L^p$-convergence of $f_r$ as below.
\begin{theorem}
Assume $1\leq p<\infty$. Let $f\in L^p(\Smone\times\Bm,\HK)$ and $g=P_B[f]$.  Then
\begin{equation} \label{grtof} \lim_{r \to 1}\|g_r-f\|_p = 0.
\end{equation}
\end{theorem}
\begin{proof} Fix $\epsilon>0$ and choose $h\in C(\Smone\times\Bm,\HK)$ such that $\|f-h\|_p\leq\epsilon$. Let $g'=P_B[h]$, then we have
\be
\|g_r-f\|_p\leq\|g_r-g'_r\|_p+\|g'_r-h\|_p+\|h-f\|_p.
\ee
Notice that $g_r-g'_r=P_B[f-h]_r$, which yields $\|g_r-g'_r\|_p\leq b_{m,k}\|f-h\|_p<b_{m,k}\epsilon$ by Theorem \ref{thmbdd}. Therefore, one has
$
\|g_r-f\|_p\leq \|g'_r-h\|_p+(b_{m,k}+1)\epsilon.
$
Since $h\in C(\Smone\times\Bm,\HK)$, Theorem \ref{Pintegral} tells us that $\|g'_r-h\|_p\rightarrow 0$ when $r\rightarrow 1$. Hence,
$
\lim_{r\rightarrow 1}\|g_r-f\|_p\leq (b_{m,k}+1)\epsilon.
$
Since $\epsilon$ is arbitrary and $b_{m,k}$ depends only on $m$ and $k$, we conclude \eqref{grtof}.
\end{proof}
With the $L^p$-convergence of $f_r$, we obtain the result of the Dirichlet problem of bosonic Laplacians in the unit ball with $L^p$ data as follows.
\begin{theorem}[Dirichlet problem in $\mathbb{B}^m$ with $L^p$ data]
Assume $1\leq p<\infty$ and $h\in L^p(\Smone\times\Bm,\HK)$. Define $h^*$ in $\mathbb{B}^m\times\mathbb{B}^m$ by
\[
h^*(\bx,\bnu)=\begin{cases}
P_B[h](\bx,\bnu),\ &\text{if}\ \bx\in\mathbb{B}^m, \quad \bnu \in \Bm, \\
h(\bx,\bnu),\ &\text{if}\ \bx\in\Smone, \quad \bnu \in \Bm.
\end{cases}
\]
Then the following hold 
\\
(1). $\Dtwo h^*=0$ in $\Bm\times\Bm$,
\\
(2). $\|h^*_r\|_p\leq b'_{m,k}\|h\|_p$, where $b'_{m,k}>0$ is a constant depending on $m$ and $k$,
\\
(3). $\lim_{r \to 1}\|h^*_r-h\|_p = 0$.
\end{theorem}
\section{Uniqueness for solutions of the Dirichlet problems}
In this section, we will apply some results of elliptic differential operators on vector bundles to prove the uniqueness for solutions to the Dirichlet problems for bosonic Laplacians in the unit ball.
\subsection{Estimates regarding linear elliptic differential operators}
\hfill\\ Let  $\zeta=(E,\pi,X)$ be a real vector bundle of rank $s$. In this section, we only consider $E=X\times\R^s$, where $\pi:\ E\longrightarrow X$ is the natural projection, and the triple $\theta_s=(E,\pi,X)$ is called the trivial bundle of rank $s$ over $X$. A bundle isomorphic to $\theta_s$ is also said to be trivial. 
\par
Now, we consider differential operators from $\theta_r$ to $\theta_s$ on an open set $\Omega\subset\R^m$. We suppose that the operator is of form
\begin{eqnarray}\label{operatorP}
P(f)(\bx)=\sum_{|\bs{\alpha}|\leq n}a_{\bs{\alpha}}(\bx)D^{\bs{\alpha}}f(\bx), f\in C^{\infty}(\Omega,r)
\end{eqnarray}
where $\bs{\alpha}=(\alpha_1,\cdots,\alpha_m)$ with $\alpha_j\geq 0,\ j=1,\cdots,m$ is a multi-index, $|\bs{\alpha}|=\sum_{j=1}^m\alpha_j$, $D^{\bs{\alpha}}:=\partial_{x_1}^{\alpha_1}\cdots\partial_{x_m}^{\alpha_m}$ and $C^{\infty}(\Omega,r)$ stands for functions in $C^{\infty}(\Omega)$ taking values in $\R^r$. If we replace $\partial_{x_j}$ with $\eta_j$, then we call 
$
p(\bx,\bs{\eta})=\sum_{|\bs{\alpha}|=n}a_{\bs{\alpha}}(\bx)\bs{\eta}^{\bs{\alpha}},\quad \bx\in\Omega,\bs{\eta}\in\R^m
$
the \emph{principal symbol} of the differential operator $P$ given in \eqref{operatorP}. A linear differential operator of order $n$ from $\theta_r$ to $\theta_s$ in $\Omega$ is called \emph{elliptic} if and only if for any $\bs{\eta}\neq0,\ \bs{\eta}\in\R^m$ and $\bx\in\Omega$, the map $p(\bx,\bs{\eta}):\R^r\longrightarrow \R^s$ is injective.
\par
Now, we introduce an estimate from \cite{Nar}, which is crucial in the proof of uniqueness for the Dirichlet problems given in the next subsection.
\begin{proposition}[3.6.9, \cite{Nar}]\label{estimates}
Let $\Omega$ be a bounded open set in $\R^m$ and $P$ a linear elliptic differential operator with constant coefficients from $\theta_r$ to $\theta_s$ of order $n$ given by 
$
Pf(\bx)=\sum_{|\bs{\alpha}|\leq n}a_{\bs{\alpha}}D^{\bs{\alpha}}f(\bx)$ for $f\in C^{\infty}(\Omega,r).
$
Let $\lambda\geq 0$  be an integer. Then there exists a constant $C>0$ such that 
\be
\|f\|_{H^{n+\lambda}}\leq C\|Pf\|_{H^\lambda},\quad\text{for\ all\ }f\in C^{\infty}_0(\Omega,r),
\ee
where $\|\cdot\|_{H^\lambda}$ is the norm of Sobolev space $H^\lambda(\Omega)$.
\end{proposition}
\begin{remark}
From the proof of $3.6.9$ in \cite{Nar}, we can replace the smoothness condition by continuity condition. 
\end{remark}
\subsection{Uniqueness for solutions of Dirichlet problems}
Recall that the bosonic Laplacian $\Dtwo$ maps $C^{\infty}(\Bm\times\Bm,\HK)$ to itself and its expression is given by \eqref{Dtwo}. In this case, we consider the function space as a trivial vector bundle $E=\Bm\times\HK$, in other words, the base $X$ is the unit ball $\Bm$ and for each point $\bx\in\Bm$, the fiber $E_{\bx}=\HK$. It is easy to observe that $\HK$ is isomorphic to $\R^t$ with $t=\dim\HK$ simply by matching the basis elements in both vector spaces. More specifically, let $\{\varphi_j(\bu)\}_{j=1}^t$ be the set of orthonormal basis for $\HK$ with respect to the $L^2$ inner product over the unit sphere, and $\{\bs{e}_j\}_{j=1}^t$ be the standard orthonormal basis of $\R^t$. If we match $\varphi_j$ with $\bs{e}_j$, $j=1,\cdots,t$, then one can see $\HK$ is isomorphic to $\R^t$ and $\|\varphi_j\|_{L^2}=|\bs{e}_j|=1,\ j=1,\cdots,t$. With this isomorphism, the ellipticity of $\Dtwo$ proved in Theorem A.$1$ \cite{DeBie} already implies the ellipticity defined in the previous subsection for differential operators on trivial bundles.
\par
Recall that the Dirichlet problem for bosonic Laplacians with continuous boundary data in the unit ball is given by 
\[
\begin{cases}
\Dtwo f(\bx,\bnu)=0,\ &\text{if}\ \bx\in\mathbb{B}^m, \quad \bnu \in \Bm,\\
f(\bx,\bnu)=g(\bx,\bnu),\ &\text{if}\ \bx\in\Smone, \quad \bnu \in \Bm,
\end{cases}
\]
where $g\in C(\Smone\times\Bm,\HK)$. Since $\Dtwo$ is linear, to prove the uniqueness for continuous solutions to the Dirichlet problem above, we only need to prove that when $g=0$, the only continuous solution is $f=0$. According to Proposition \ref{estimates}, choosing $\lambda=0$, we have
$
\|f\|_{H^{2}}\leq C\|\Dtwo f\|_{L^2}=0,
$
which tells us that $f=0$ almost everywhere in $\Bm\times\Bm$. Further, the continuity of $f$ in $\Bm\times\Bm$ immediately gives us that $f=0$ in $\Bm\times\Bm$. This completes the proof of the uniqueness for solutions to the Dirichlet problem in the unit ball.
\section{Properties of null solutions of bosonic Laplacians}
In this section, we will use the uniqueness of continuous extension of Theorem \ref{Pintegral} to obtain analogs of several results of harmonic functions for null solutions of bosonic Laplacians when $m>4$.
\subsection{Mean value property}
Since the Poisson integral formula for a harmonic function $f$ is given by
$
f(\bx)=\int_{\Smone}\frac{1-|\bx|^2}{|\bx-\ze|^2}f(\ze)dS(\ze),\quad \bx\in\Bm,
$
one can obtain the mean-value property of $f$ over the unit sphere by letting $\bx=0$. This motivates us that choosing $\z=0$ in the Poisson integral given in Theorem \ref{Pintegral} should give us a mean-value property for null solutions of bosonic Laplacians $\Dtwo$ as well. It is worth pointing out that the mean-value property discussed below is with respect to the variable $\bx$, but $f(\bx,\bu)$ also has a mean-value property with respect to $\bu$. This is because $f$ is harmonic with respect to $\bu$ in accordance to the definition of the function space $C^2(\Bm\times\Bm,\HK)$.
\begin{theorem}[Mean-value property: sphere version]\label{prop1}
\hfill\\
Assume $f\in C^2(B(\bs{a},r)\times\Bm,\HK)\cap C(\overline{B(\bs{a},r)}\times\overline{\Bm},\HK)$ and  $\Dtwo f=0$ in ${B(\bs{a},r)}\times{\Bm}$. We have
\be
f(\bs{a},\bnu)=\frac{c_{m,k}}{2}\int_{\Smone}f(\bs{a}+r\ze,\ze\bnu\ze)dS(\ze),\quad  \forall\bnu\in {\Bm},
\ee
where $c_{m,k}$ is given in Lemma \ref{lemma1}.
\end{theorem}
\begin{remark}
We notice that when $f(\bs{a},\bnu)=f(\bnu)$, the mean-value property above reduces to \cite[Lemma 6]{Dunkl}.
\end{remark}
\begin{proof}
Without loss of generality, we assume that $B(\bs{a},r)=\Bm$. Then, we can obtain the result in $B(\bs{a},r)$ by applying a translation and a dilation. Theorem \ref{Pintegral} and the uniqueness property tell us that
\be
f(\z,\bnu)=\int_{\Smone}\int_{\Smone}P_B(\bx,\ze,\bom,\bnu)f(\ze,\bom)dS(\bom)dS(\ze),
\ee
where $\z,\bnu\in \mathbb{B}^m$. Let $\z=0$, then we have
\begin{align*}
f(0,\bnu)&=\frac{c_{m,k}}{2}\int_{\Smone}\int_{\Smone}Z_k(\ze\bs{\omega}\ze,\bnu)f(\ze,\bom)dS(\bom)dS(\ze)\\
&=\frac{c_{m,k}}{2}\int_{\Smone}\int_{\Smone}Z_k(\bs{\omega},\ze\bnu\ze)f(\ze,\bom)dS(\bom)dS(\ze)\\
&=\frac{c_{m,k}}{2}\int_{\Smone}f(\ze,\ze\bnu\ze)dS(\ze),
\end{align*}
where we have used the facts that $Z_k(\bs{a}\bu\bs{a},\bv)=Z_k(\bu,\bs{a}\bv\bs{a})$ for $\bs{a}\in\R^m$ and that $Z_k$ is the reproducing kernel of $\mathcal{H}_k$ in the last two steps.
\end{proof}
Further, we can also obtain a volume version of the mean-value property by changing to an integral over the sphere and the radius.
\begin{proposition}[Mean-value property: volume version]\label{prop2}
\hfill\\
Assume $f\in C^2(B(\bs{a},r)\times\Bm,\HK)\cap C(\overline{B(\bs{a},r)}\times\overline{\Bm},\HK)$ and  $\Dtwo f=0$ in ${B(\bs{a},r)}\times{\Bm}$. We have
\be
f(\bs{a},\bnu)=\frac{m+2k-2}{(m-2)V(B(\bs{a},r))}\int_{B(\bs{a},r)}f\bigg(\z,\frac{\bs{\eta}\bnu\bs{\eta}}{|\bs{\eta}|^2}\bigg)d\bx, \quad \forall \bnu\in {\Bm},
\ee
where $\bs{\eta}=\frac{\z-\bs{a}}{|\z-\bs{a}|}$ and $V(B(\bs{a},r))$ denotes the volume of the ball $B(\bs{a},r)$.
\end{proposition}

\subsection{Cauchy's estimates and Liouville-type theorem} 
\hfill\\
Let $\bs{\alpha}=(\alpha_1,\cdots,\alpha_m)$ be a multi-index, where $\alpha_j,\ j=1,\cdots, m$ are non-negative integers. It is easy to see that for each $\ze,\bom\in\Smone$,  the Poisson kernel $P_B(\z,\ze,\bom,\bnu)$ given in \eqref{pkernel}
 is infinitely differentiable in $\mathbb{B}^m$ with respect to $\z$, where $\z,\bnu\in\mathbb{B}^m$. Further, since there are no singular points occur for $P_B(\z,\ze,\bom,\bnu)$ when $\z\in\Bm$ and $\ze\in\Smone$, one observes that $D^{\bs{\beta}}_{\bv}D^{\bs{\alpha}}_{\bx}P_B(\z,\ze,\bom,\bnu)$ are integrable over $\Bm\times\Bm$, where $\bs{\alpha},\ \bs{\beta}$ are multi-indices. \par

Recall that if $f\in C^2(\mathbb{B}^m\times\mathbb{B}^m,\HK)\cap C(\overline{\mathbb{B}^m}\times\overline{\mathbb{B}^m},\HK)$ and $\Dtwo f=0$, then Theorem \ref{Pintegral}  and the uniqueness for solutions of the Dirichlet problems yield
\be
f(\z,\bnu)=\int_{\Smone}\int_{\Smone}P_B(\z,\ze,\bom,\bnu)f(\ze,\bom)dS(\bom)dS(\ze),\quad \bx,\bnu\in\Bm.
\ee
Differentiating under the integral sign, we can see that $f(\z,\bnu)$ is infinitely differentiable with respect to $\z,\bnu$ in $\mathbb{B}^m$. This is an analog of the smoothness property of harmonic functions. Further, we have
\be
D^{\bs{\beta}}_{\bnu}D_{\z}^{\bs{\alpha}}f(\z,\bnu)=\int_{\Smone}\int_{\Smone}f(\ze,\bom)D^{\bs{\beta}}_{\bnu}D_{\z}^{\bs{\alpha}}P_B(\z,\ze,\bom,\bnu)dS(\bom)dS(\ze).
\ee
We can derive an analog of the Cauchy's estimates as follows.
\begin{theorem}[Cauchy's estimates]\label{CE0}
Let $\bs{\alpha},\ \bs{\beta}$ be multi-indices. Assume $f\in C^2(\Omega\times \Bm,\mathcal{H}_k)$ and $\Dtwo f=0$ in $\Omega\times\Bm$. Then there exists a constant $c_{\bs{\alpha},m,k}$ such that
\begin{equation*}
|D_{\bnu}^{\bs{\beta}}D_{\z}^{\bs{\alpha}}f(\bs{a},\bnu_0)|\leq\frac{c_{\bs{\alpha},m,k}\|f\|_{L^{\infty}(B(\bs{a},r_1)\times B(\bnu_0,r_2),\HK)}}{r_1^{|\bs{\alpha}|}r_2^{|\bs{\beta}|}},
\end{equation*}
for any $B(\bs{a},r_1)\Subset \Omega$ and  $B(\bnu_0,r_2)\Subset \Bm$.
\end{theorem}
\begin{proof}
Without loss of generality, we may assume $\bs{a}=0$. Firstly, we consider the case  $\Omega =\Bm$ and let $\|f\|_{L^{\infty}(\Bm\times \Bm,\HK)}=M'$ , then for $\bnu_0 \in \Bm$,
\begin{align*}
|D_{\bnu}^{\bs{\beta}}D_{\z}^{\bs{\alpha}}&f(0,\bnu_0)|=\bigg\rvert\int_{\Smone}\int_{\Smone}f(\ze,\bom)D_{\bnu}^{\bs{\beta}}D_{\z}^{\bs{\alpha}}P_B(0,\ze,\bom,\bnu_0)dS(\bom)dS(\ze)\bigg\rvert\\
&\leq M'\int_{\Smone}\int_{\Smone}\big\rvert D_{\bnu}^{\bs{\beta}}D_{\z}^{\bs{\alpha}}P_B(0,\ze,\bom,\bnu_0)\big\rvert dS(\bom)dS(\ze)\\
&=M'\int_{\Smone}\int_{\Smone}\bigg\rvert D_{\bnu}^{\bs{\beta}}\big\rvert_{\bnu=\bnu_0}D_{\z}^{\bs{\alpha}}\big\rvert_{\bx=0}P_B(\bx,\ze,\bom,\bnu)\bigg\rvert dS(\bom)dS(\ze).
\end{align*}
Since there are no singular points of the function
$P_B(\bx,\ze,\bom,\bnu),$ which is a homogeneous polynomial with respect to $\bnu$ (see page 104 in \cite{Axler}) where $\bom\in\Smone,\ \bnu\in\Bm$. This implies that 
$$
\left| D_{\bnu}^{\bs{\beta}}\big\rvert_{\bnu=\bnu_0}D_{\z}^{\bs{\alpha}}\big\rvert_{\bx=0}P_B(\bx,\ze,\bom,\bnu)\right|
$$
is bounded when $\z=0,\ \bom\in\Smone,\ \ze\in\Smone$ and $\bnu=\bnu_0$. Therefore, there exists a constant $c_{\bs{\alpha},m,k}$, which only depends on $\bs{\alpha},m$ and $k$, such that
\be
\int_{\Smone}\bigg\rvert D_{\bnu}^{\bs{\beta}}\big\rvert_{\bnu=\bnu_0}D_{\z}^{\bs{\alpha}}\big\rvert_{\bx=0}P_B(\bx,\ze,\bom,\bnu)\bigg\rvert dS(\bom)dS(\ze)\leq c_{\bs{\alpha},m,k}.
\ee
Hence, we have
$
|D_{\bnu}^{\bs{\beta}}D_{\bx}^{\bs{\alpha}}f(0,\bnu_0)|\leq c_{\bs{\alpha},m,k}M'.
$
Now, suppose $\Dtwo f=0$ and $\|f\|_{L^{\infty}(B(0,r_1)\times B(\bnu_0,r_2),\HK)}=M$. Then we apply the argument above to $f(r_1\z,r_2\bnu+\bnu_0)$ with respect to $\z,\ \bnu\in\Bm$, we obtain
\be
|D_{\bnu}^{\bs{\beta}}D_{\bx}^{\bs{\alpha}}f(0,\bnu_0)|\leq \displaystyle\frac{c_{\bs{\alpha},m,k}'M}{r_1^{|\bs{\alpha}|}r_2^{|\bs{\beta}|}}.
\ee
The proof is complete.
\end{proof}
One can also obtain Cauchy's estimates with an $L^1$ norm as follows.
\begin{proposition}\label{CE}
Let $\bs{\alpha},\ \bs{\beta}$ be multi-indices. Suppose $f\in C^2(\Omega\times\mathbb{B}^m,\mathcal{H}_k)$, $\Dtwo f=0$ in $\Omega\times\Bm$. Then there exists a constant $c_{\bs{\alpha},m,k}$ such that for any $B(\bs{a},r_1)\Subset \Omega$ and $\bnu_0 \in \Bm$, there holds
\be
|D_{\bnu}^{\bs{\beta}}D_{\z}^{\bs{\alpha}}f(\bs{a},\bnu_0)|\leq\frac{c_{\bs{\alpha},m,k}\|f\|_{L^1(B(\bs{a},r_1)\times B(0,1-\frac{1}{4} r_2),\HK)}}{r_1^{m+|\bs{\alpha}|}{r_2}^{m+|\bs{\beta}|}},
\ee
where $r_2$ stands for the distance from $\bnu_0 \in \Bm$ to the unit sphere $\Smone$ and
\be
\|f\|_{L^1(B(\bs{a},r_1)\times B(0,1-\frac{1}{4}r_2),\HK)}:=\int_{B(\bs{a}_1,r_1)}\int_{B(0,1-\frac{1}{4}r_2)}|f(\bx,\bnu)|d\bnu d\bx.
\ee
\end{proposition}
\begin{proof} Let $B(\bs{a},r_1)\Subset \Omega$ and $\bnu_0 \in \Bm$. Without loss of generality, we may assume that $\bs{a}=0$. Put $M=\|f\|_{L^{\infty}(B(0,r_1/2)\times B(\bnu_0,r_2/2),\HK)}<\infty$.  Then we apply the argument in the theorem above to $f(r_1\z/2,r_2\bnu/2+\bnu_0)$ with respect to $\z,\ \bnu\in\Bm$ to have
\begin{eqnarray}\label{Cestimate}
|D_{\bnu}^{\bs{\beta}}D_{\bx}^{\bs{\alpha}}f(0,\bnu_0)|\leq \frac{c_{\bs{\alpha},m,k}M}{r_1^{|\bs{\alpha}|}r_2^{|\bs{\beta}|}}.
\end{eqnarray}
Next, we assume that $f$ obtains its maximum value $M$ over $\overline{B(0,r_1/2)}\times\overline{B(\bnu_0,r_2/4)}$ at the point $(\bx_1,\bnu_1)$. We can see that $$B(\bx_1,\frac{r_1}{2})\subset B(0,r_1),\ B(\bnu_1,\frac{r_2}{2})\subset B(\bnu_0,\frac{3r_2}{4}).$$ Therefore, in accordance to the mean-value property, we have
\be
M=f(\bx_1,\bnu_1)=\frac{m+2k-2}{(m-2)V(B(\bs{x}_1,\frac{r_1}{2}))}\int_{B(\bs{x}_1,\frac{r_1}{2})}f\bigg(\bx,\frac{\bs{\eta}\bnu_1\bs{\eta}}{|\bs{\eta}|^2}\bigg)d\bx.
\ee
Notice that if we let $\frac{\bs{\eta}\bnu_1\bs{\eta}}{|\bs{\eta}|^2}=\bnu_2$, which means that $\bnu_2$ is obtained from $\bnu_1$ by a rotation, then $f(\bx,\bnu_2)\in\HK$ with respect to $\bnu_2\in \Bm$. Further, one can notice that $B(\bnu_2,\frac{r_2}{2})\subset B(0,1-\frac{1}{4}r_2)$. Hence, one can use the mean-value property of harmonic functions to obtain that
\begin{align*}
|f(\bx,\bnu_2)|=\frac{1}{V(B(\bnu_2,\frac{r_2}{2}))}\bigg\rvert\int_{B(\bnu_2,\frac{r_2}{2})}f(\bx,\bnu)d\bnu\bigg\rvert \leq \frac{c_1}{r_2^{m}}\int_{B(0,1-\frac{1}{4} r_2)}|f(\bx,\bnu)|d\bnu, 
\end{align*}
where $\ c_1$ is a positive constant only depending on $m$. Therefore, we have 
\begin{align*}
M=&f(\bx_1,\bnu_1)\leq \frac{c'}{r_2^{m}V(B(\bs{x}_1,\frac{r_1}{2}))}\int_{B(\bs{x}_1,r_1/2)}\int_{B(0,1-\frac{1}{4} r_2)}|f(\bx,\bnu)|d\bnu d\bx\\
&\leq c''r_1^{-m}r_2^{-m}\|f\|_{L^1(B(0,r_1)\times B(0,1-\frac{1}{4} r_2),\HK)},
\end{align*}
where the last inequality comes from the fact that $B(\bx_1,r_1/2)\subset B(0,r_1)$ and $c',c''$ are positive constants only depending on $m$ and $k$. Combining with (\ref{Cestimate}) completes the proof.
\end{proof}
\begin{remark}
In this proposition, one might notice that the domain of $\bnu$ in the $L^1$ norm is $B(0,1-\frac{1}{4} r_2)$ instead of $B(\bnu_0,r)$, which is the form of the domain of $\bs{a}$. This is because of the difficulty caused by $\bnu_2$ in the proof above. Actually, $\bnu_2$ is obtained by a rotation of $\bnu_1$, and it can be anywhere on the sphere with the radius $|\bnu_1|$. Hence, it seems impossible to cover $\bnu_2$ with a ball $B(\bnu_0,r)$ by a similar argument applied for $\bs{a}$.  Also, the $L^1$ norm defined here is equivalent to the norm defined in (\ref{norm}) up to a constant with similar argument given in \eqref{equivalence}. This constant only depends on $m$ and $k$ considering $3/4\leq1-\frac{1}{4} r_2\leq 1$, since $0\leq r_2\leq 1$.
\end{remark}
One can apply Cauchy's estimates to obtain an analog of Liouville's Theorem for null solutions of bosonic Laplacians as follows.
\begin{theorem}[Liouville-type Theorem]\label{LT}
 Suppose $f\in C^2(\R^m\times\mathbb{B}^m,\HK) \cap L^\infty(\R^m\times\mathbb{B}^m,\HK)$ and $\Dtwo f=0$ in $\R^m\times\mathbb{B}^m$. Then $f=f(\bnu)\in\HK$.
\end{theorem}
\begin{proof}
Suppose $f$ is bounded by $M$ in $\R^m\times\Bm$. Let $B(\bs{a},R)$ be an arbitrary ball, then from the proof of the Cauchy's estimates above, one can immediately obtain that
$
|\nabla_{\bx}f(\bs{a},\bnu)|\leq \frac{c_{m,k}M}{R}.
$
Since $M$ does not depend on $R$, we let $R\rightarrow \infty$, which gives us that $|\nabla_{\bx}f(\bs{a},\bnu)|=0$. Hence, $f=f(\bnu)\in\HK$.
\end{proof}
An immediate consequence of the theorem given above is the following.
\begin{proposition} Let $1\leq l \in {\mathbb N}$. Suppose $f\in C^2(\R^m\times\mathbb{B}^m,\HK)$, $\Dtwo f=0$ in $\R^m\times\mathbb{B}^m$ and
$$\|f\|_{L^1(B(0,R)\times\Bm,\HK)}=o(R^{l+m})  \quad \text{as } R\rightarrow \infty.
$$
 Then $f$ is a polynomial of $\bx$ with degree less than $l$.
\end{proposition}
\begin{proof}
Let $(\bs{a},\bnu)$ be an arbitrary point in $\R^m\times\Bm$, and $B(\bnu,r_2)\subset\Bm$. Now, we denote $|\bs{a}|=r_1$, then one can see that $B(\bs{a},R)\subset B(0,R+2r_1)$. In accordance to the Cauchy's estimates, we have, for any multi-index $\bs{\alpha}$, 
\begin{align*}
|D_{\bx}^{\bs{\alpha}}f(\bs{a},\bnu)|&\leq\frac{c_{\bs{\alpha},m,k}\|f\|_{L^1(B(\bs{a},R)\times \Bm,\HK)}}{R^{m+|\bs{\alpha}|}r_2^{m}}\leq \frac{c_{\bs{\alpha},m,k}\|f\|_{L^1(B(0,R+2r_1)\times \Bm,\HK)}}{R^{m+|\bs{\alpha}|}r_2^{m}}\\
&=c_{\bs{\alpha},m,k}\frac{o((R+2r_1)^{l+m})}{R^{m+|\bs{\alpha}|}r_2^{m}},\ \text{as } R\rightarrow\infty.
\end{align*}
Since $r_1,r_2$ are fixed, and $c_{\bs{\alpha},m}$ only depends on $\bs{\alpha},m$, if we let $|\bs{\alpha}|=l$ and $R\rightarrow\infty$, we have $|D_{\bx}^{\bs{\alpha}}f(\bs{a},\bnu)|=0$ for any $|\bs{\alpha}|=l$. Further, since $(\bs{a},\bnu)$ is arbitrary in the domain, it implies that $|D_{\bx}^{\bs{\alpha}}f(\bx,\bnu)|=0$ for any $|\bs{\alpha}|=l$ in $\R^m\times\Bm$. Hence, $f$ is a polynomial of $\bx$ with degree less than $l$.
\end{proof}

The next result provides the convergence property.  
\begin{proposition}
Suppose that there is a sequence of functions $\{f_j(\bx,\bnu)\}\subset C^2(\Omega\times\Bm,\HK)$ such that $\Dtwo f_j=0$ for any $1 \leq j \in {\mathbb N}$. Further, assume that  $\{\|f_j\|_{L^1(K\times \Bm,\HK)}\}_{j=1}^{\infty}$ is uniformly bounded, where $K$ is an arbitrary compact subset of $\Omega$. Then, there exists a subsequence $\{f_{j_k}\}$ and a function $f(\bx,\bnu)\in C^2(\Omega\times\Bm,\HK)$  such that for every multi-index $\bs{\alpha}$ and any compact subset $K_1\times K_2$ of $\Omega\times\Bm$, $\{ D_{\bx}^{\bs{\alpha}}f_{j_k}\}$ converges uniformly to  $D_{\bx}^{\bs{\alpha}}f$ in $K_1\times K_2$. Moreover, $f\in C^2(\Omega\times\Bm,\HK)$ and $\Dtwo f=0$ in $\Omega\times\Bm$.
\end{proposition}
\begin{proof}
Here, we only need to show that there exist a subsequence $\{f_{j_k}\}$ and a function $f(\bx,\bnu)\in C^2(\Omega\times\Bm,\HK)$  such that  for each multi-index $\bs{\alpha}$ and each compact subset $K_1\times K_2$ of $\Omega\times\Bm$, $\{D^{\bs{\alpha}}f_j\}$ converges uniformly to  $D^{\bs{\alpha}}f$ in $K_1\times K_2$ and $\Dtwo f=0$ in $K_1\times K_2$.
\par
Let $\bs{a}$ be an arbitrary point in $\Omega$, and $B(\bs{a},2r_1)\subset\Omega,\ B(0,2r_2)\subset\Bm$. We choose $K=\overline{B(\bs{a},2r_2)},\ K_1=\overline{B(\bs{a},r_1)},\ K_2=\overline{B(0,r_2)}$. According to the assumption, $\{\|f_j\|_{L^1(K\times\Bm,\HK)}\}$ is uniformly bounded by some $M>0$. Then, in accordance to the Cauchy's estimates given in Theorem \ref{CE0}, we have that for all $(\bx,\bnu)\in K_1\times K_2$,
$
|D_{\z}^{\bs{\alpha}}f_j(\bx,\bnu)|\leq\frac{c_{\bs{\alpha},m,k}M}{r_1^{m+|\bs{\alpha}|}r_2^m}.
$
 By choosing $|\bs{\alpha}|=0,1$ above respectively, one can see that $\{f_j\}$ is uniformly bounded and uniformly equicontinuous in $K_1\times K_2$, Arzel\`a-Ascoli Theorem tells us that there exists a subsequence $\{f_{j_k}\}$ converges uniformly to a function $f\in C^2(K_1\times K_2,\HK)$ in $K_1\times K_2$. Since $\bs{a},\ K_1,\ K_2$ are arbitrary, we conclude that $f\in C^2(\Omega\times\Bm,\HK)$.
\par
Further, $\Dtwo f_{j_k}=0$ tells us that
\be
f_{j_k}(\z,\bnu)=\int_{\Smone}\int_{\Smone}P_B(\bx,\ze,\bom,\bnu)f_{j_k}(\ze,\bom)dS(\bom)dS(\ze),
\ee 
for every integer $j_k$, and $\z,\ \bnu\in\Bm$. Since $\{f_{j_k}\}$ is uniformly bounded in every compact subset $K_1\times K_2\subset\Omega\times\Bm$, we can take the limit on both sides to obtain
\be
f(\z,\bnu)=\int_{\Smone}\int_{\Smone}P_B(\bx,\ze,\bom,\bnu)f(\ze,\bom)dS(\bom)dS(\ze)
\ee
for every $\z\in K_1,\ \bnu\in K_2$, hence, $\Dtwo f=0$ in $K_1\times K_2$. For a multi-index $\bs{\alpha}$, we have
\be
D_{\bx}^{\bs{\alpha}}f_{j_k}(\z,\bnu)=\int_{\Smone}\int_{\Smone}D_{\bx}^{\bs{\alpha}}P_B(\bx,\ze,\bom,\bnu)f_{j_k}(\ze,\bom)dS(\bom)dS(\ze).
\ee
Since $Z_k(\bom,\bnu)$ is a degree-$k$ homogeneous polynomial in $\bom$ and $\bnu$, we observe that for every compact subset $K_1\times K_2\subset\Omega\times\Bm$, 
$
D_{\bx}^{\bs{\alpha}}P_B(\bx,\ze,\bom,\bnu)
$
is uniformly bounded in $(\z,\ze,\bom,\bnu)\in K_1\times\Smone\times\Smone\times K_2$. Thus
\begin{align*}
&\lim_{j\rightarrow \infty}D_{\bx}^{\bs{\alpha}}f_{j_k}(\z,\bnu)
=\lim_{j\rightarrow \infty}\int_{\Smone}\int_{\Smone}D_{\bx}^{\bs{\alpha}}P_B(\bx,\ze,\bom,\bnu)f_{j_k}(\ze,\bom)dS(\bom)dS(\ze)\\
&=D_{\bx}^{\bs{\alpha}}f(\z,\bnu),
\end{align*}
which completes the proof.
\end{proof}
\bibliographystyle{siam}

\begin{thebibliography}{99}


\bibitem{Axler} S. Axler, P. Bourdon, W. Ramey, \emph{Harmonic function theory}, second edition, Graduate Texts in Mathematics, Springer, New York, 2001.

\bibitem{Ban} R. Ba\~{n}uelos, C.N. Moore, \emph{Probabilistic Behavior of Harmonic Functions}, Birkh\"{a}user, 1991.

\bibitem{Bargmann} V. Bargmann, E.P. Wigner, \emph{Group theoretical discussion of relativistic wave equations}, Proceedings of the National Academy of Sciences of the United States of America. 34(5), 1948, pp. 211-223.

\bibitem{Bures} J. Bure\v{s}, F. Sommen, V. Sou\v{c}ek, P. Van Lancker, \emph{Rarita-Schwinger Type Operators in Clifford Analysis}, J. Funct. Anal. Vol. 185, No. 2, 2001, pp. 425-455.

\bibitem{CO} J.L. Clerc, B. \O rsted, \emph{Conformal covariance for the powers of the Dirac operator}, https://arxiv.org/abs/1409.4983v1.

\bibitem{DeBie} H. De Bie, D. Eelbode, M. Roels, \emph{The Higher Spin Laplace Operator}, Potential Analysis, Vol. 47, Issue 2, 2017, pp. 123-149.


\bibitem{Del}R. Delanghe, F. Sommen, V. Sou\v{c}ek, \emph{Clifford Algebra and Spinor-Valued Functions, A Function Theory for the Dirac Operator}, Springer Netherlands, 1992.

\bibitem{Ding} C. Ding, \emph{Integral Formulas for Higher Order Conformally Invariant Fermionic Operators}, Advances in Applied Clifford Algebras, 29, 37 (2019) doi:10.1007/s00006-019-0953-4.

\bibitem{DR} C. Ding, J. Ryan, \emph{Some Properties of the Higher Spin Laplace Operator}, Transactions of the American Mathematical Society, Vol. 371, Issue 5, 2019, pp. 3375-3395.


\bibitem{DWR}C. Ding, R. Walter, J. Ryan, \emph{Construction of Arbitrary Order Conformally Invariant Operators in Higher Spin Spaces}, The Journal of Geometric Analysis, Vol. 27, Issue 3, 2017, pp. 2418-2452.


\bibitem{Dunkl} C.F. Dunkl, J. Li, J. Ryan, P. Van Lancker, \emph{Some Rarita-Schwinger type operators}, Computational Methods and Function Theory, Vol. 13, Issue 3, 2013, pp. 397-424.

\bibitem{Eelbode} D. Eelbode, M. Roels, \emph{Generalised Maxwell equations in higher dimensions}, Complex Analysis and Operator Theory, Vol. 10, Issue 2, 2016, pp. 267-293.


\bibitem{Gilbert} J. Gilbert, M. Murray, \emph{Clifford Algebras and Dirac Operators in Harmonic Analysis},  Cambridge University Press, Cambridge, 1991.








\bibitem{Nar} R. Narasimhan, \emph{Analysis on Real and Complex Manifolds}, North-Holland Mathematical Library, Vol. 35, 1985.

\bibitem{Porteous} I. Porteous, \emph{Clifford algebra and the classical groups}, Cambridge University Press, Cambridge, 1995.

\bibitem{JR} J. Ryan, \emph{Dirac Operators, Conformal Transformations and Aspects of Classical Harmonic Analysis}, Journal of Lie Theory, Vol. 8, 1998, pp. 67-82.

\bibitem{Stein} E. Stein, G. Weiss, \emph{Generalization of the Cauchy-Riemann equations and representations of the rotation group}, Amer. J. Math. 90 (1968), pp. 163-196.




\end{thebibliography}

\end{document}